\documentclass[%
 preprint,
nofootinbib,
 amsmath,amssymb,
 aps,
]{revtex4-2}

\usepackage{graphicx}
\usepackage{dcolumn}
\usepackage{bm}
\usepackage{graphicx}
\usepackage{physics}    
\usepackage{hyperref}   

\usepackage{amsmath}    
\usepackage{amssymb}    
\usepackage{bm}         
\usepackage{physics}    
\usepackage{amsfonts}
\newtheorem{proposition}{Proposition}
\newtheorem{proof}{Proof}
\usepackage{bm}
\usepackage{braket}
\usepackage{float} 

\begin{document}


\title{Quantum Natural Gradient with Geodesic Corrections for Small Shallow Quantum Circuits}

\author{Mourad Halla }

\affiliation{Deutsches Elektronen-Synchrotron DESY, Platanenallee 6, 15738 Zeuthen, Germany}%


\begin{abstract}
The Quantum Natural Gradient (QNG) method enhances optimization in variational quantum algorithms (VQAs) by incorporating geometric insights from the quantum state space through the Fubini-Study metric. In this work, we extend QNG by introducing higher-order integrators and geodesic corrections using the Riemannian Euler update rule and geodesic equations, deriving an updated rule for the Quantum Natural Gradient with Geodesic Correction (QNGGC). We also develop an efficient method for computing the Christoffel symbols necessary for these corrections, leveraging the parameter-shift rule to enable direct measurement from quantum circuits. Through theoretical analysis and practical examples, we demonstrate that QNGGC significantly improves convergence rates over standard QNG, highlighting the benefits of integrating geodesic corrections into quantum optimization processes. Our approach paves the way for more efficient quantum algorithms, leveraging the advantages of geometric methods.
\end{abstract}

\maketitle


\section{Introduction}

Quantum computing represents a significant leap in computational science, enabling the resolution of problems that are fundamentally intractable for classical algorithms by leveraging quantum mechanical principles. Among the leading strategies in near-term quantum computing are Variational Quantum Algorithms (VQAs) \citep{Cerezo2021, McClean2016, Bharti2022}, which are tailored to harness the power of current noisy intermediate-scale quantum (NISQ) devices. A prominent example of VQAs is the Variational Quantum Eigensolver (VQE) \cite{Peruzzo2014}, a hybrid quantum-classical algorithm designed to find the ground state energy of quantum systems a task of paramount importance in areas such as quantum chemistry, materials science and condensed matter physics. VQE operates by iteratively optimizing a parameterized quantum state, known as an ansatz, using a combination of quantum and classical computations. Quantum processors are employed to prepare and measure the quantum state, while classical optimization algorithms adjust the parameters of the ansatz to minimize the cost function. This synergy between quantum evaluations and classical optimization enables VQE to efficiently explore the solution space, making it a powerful tool for solving complex problems such as those found in High-Energy Physics applications \cite{karl}.

Optimization techniques are crucial to the performance of VQAs, as they directly influence convergence rates and the quality of the solutions obtained. While traditional methods like vanilla gradient descent (GD) are commonly used due to their straightforward implementation, the unique challenges of quantum optimization landscapes characterized by non-convexity, noise, and barren plateaus demand more sophisticated techniques to enhance convergence and performance.

The Quantum Natural Gradient (QNG) algorithm \cite{Stokes2020}, a generalization of the natural gradient \cite{Amari1998}, is an advanced optimization technique that enhances the optimization process by incorporating the geometry of the parameter space into the update rules. Unlike standard gradient descent, which assumes a flat parameter space, QNG leverages a Riemannian metric defined by the Fubini–Study metric or, more generally, by the Quantum Fisher Information Matrix (QFIM) \cite{Meyer}, capturing the infinitesimal changes in distances between quantum states and aligning parameter updates with the natural curvature of the quantum state manifold. A key factor in QNG’s superior performance, as shown in \cite{Katabarwa2022}, is its ability to identify regions of high negative curvature early in the optimization process, significantly accelerating convergence. These regions play a crucial role in guiding the optimization along paths that lead to faster descent.

Following the foundational work by \cite{Stokes2020}, several extensions of the QNG have been developed to broaden its applicability and robustness. For example, the QNG was extended to handle noisy and nonunitary circuits, making it more suitable for realistic quantum devices that are subject to imperfections and decoherence \cite{Koczor}. Other works propose using simultaneous perturbation stochastic approximation techniques to approximate the QFIM in QNG \cite{Gacon}. Recently, \cite{Kolotouros} introduced two methods that significantly reduce the resources needed for state preparations required for QNG: the Random Natural Gradient and the Stochastic-Coordinate Quantum Natural Gradient.

The design of the ansatz, which serves as the parameterized quantum circuits in VQAs, plays a pivotal role in the algorithm’s success. The choice of ansatz directly influences the expressiveness and trainability of the quantum circuit, as well as its efficiency in hardware simulations, thus significantly impacting the overall performance of the VQA \cite{Kandala2017, Sim2019Advanced, Tobias}. A well-designed ansatz must strike a balance between complexity and the capacity to accurately represent the target state, making the selection of the ansatz a crucial factor in the effective implementation of VQAs.

In QNG, Riemannian geometry \cite{Lee2013, Wald1984, Frankel2011} provides the framework to understand the geometric structure of the parameter space. The Fubini-Study metric, after Tikhonov-Regularization due to its general ill-definition in VQA, is a Riemannian metric that defines the intrinsic geometry of the quantum state space, guiding the optimization process by aligning parameter updates with the manifold’s curvature. Christoffel symbols describe how directions change as vectors are transported along this curved space, while geodesics represent the shortest and most efficient paths. Accurately computing these geodesics is essential for incorporating higher-order curvature corrections in QNG, enhancing optimization precision and overall performance.

The concept of using geodesic corrections to enhance optimization on manifolds was first introduced in classical optimization \cite{Transtrum2011} and further developed in subsequent work \cite{Transtrum2012}, primarily applied to nonlinear least squares problems under specific curvature assumptions. These methods underscored the crucial influence of manifold geometry in enhancing optimization performance. More recently, geodesic corrections were incorporated into natural gradient optimization within classical machine learning, demonstrating improved convergence by preserving higher-order invariance properties \cite{Song2018}. The formalism presented in Ref. \cite{Song2018} is related to our approach. Building on these foundational ideas, our work integrates geodesic corrections into the Quantum Natural Gradient (QNG), specifically designed for variational quantum algorithms. This approach exploits the unique geometric properties of the quantum state space, enabling optimization that respects and utilizes the manifold's inherent curvature for improved performance.

This manuscript presents a comprehensive approach to incorporating geodesic corrections into the QNG for applications in variational quantum algorithms. In Section~\ref{Differential}, we provide an overview of differential geometry and geodesic equations, establishing the foundational mathematical context. Section~\ref{Circuits} delves into optimizing idealized variational quantum circuits using Quantum Natural Gradient Descent (QNG), emphasizing its effectiveness in quantum optimization. Section~\ref{Higher} introduces higher-order integrators and derives the update rule for the Quantum Natural Gradient with Geodesic Correction (QNGGC). As the update rule relies on the Christoffel symbols of the second kind, Section~\ref{Christoffel} is dedicated to efficiently computing these symbols using the parameter-shift rule, enabling direct measurements from quantum circuits. In Section~\ref{Examples}, we apply the derived update rule to various examples: Examples 1 and 2 involve analytical calculations for numerical simulations, while Example 3 utilizes quantum software, specifically Qiskit \cite{Qiskit}, for practical simulations. Finally, in Section~\ref{Conclusions}, we provide an outlook on potential extensions and future research directions, highlighting the integration of geodesic corrections into quantum optimization frameworks.

\section{Differential Geometry and Geodesic Equations}
\label{Differential}

To establish our notation, we will briefly review essential concepts in differential geometry relevant to our work. For readers interested in further details, we recommend \cite{Wald1984} and \cite{Frankel2011} for introductions tailored to physicists, and \cite{Lee2013} for more comprehensive technical discussions.

A \textit{manifold} is a topological space that locally resembles Euclidean space and supports a consistent coordinate system. More formally, an \(n\)-dimensional manifold \(\mathcal{M}\) is a set equipped with a collection of coordinate charts \(\{(U_i, \varphi_i)\}_{i \in I}\), where \(I\) is an index set and each \(U_i \subset \mathcal{M}\) is an open subset
, and \(\varphi_i: U_i \rightarrow \mathbb{R}^n\) is a homeomorphism, meaning that \(U_i\) is locally similar to \(\mathbb{R}^n\). For \(\mathcal{M}\) to be a \textit{differentiable manifold}, the transition maps between overlapping charts, \(\varphi_j \circ \varphi_i^{-1}: \varphi_i(U_i \cap U_j) \rightarrow \varphi_j(U_i \cap U_j)\), must be smooth (infinitely differentiable). This smooth structure allows us to perform calculus on the manifold.

A \textit{Riemannian manifold} is a differentiable manifold \(\mathcal{M}\) equipped with a \textit{metric tensor} \(g_{ij}\), which is a symmetric, positive-definite tensor field assigning an inner product to each tangent space \(T_p\mathcal{M}\) at a point \(p \in \mathcal{M}\). In local coordinates \(\{x^i\}\), the metric tensor defines the line element:
\begin{equation}
ds^2 = g_{ij} \, dx^i \, dx^j,
\end{equation}
where $g_{ij} = g_{ji}$, and the indices $i,j$ run from $1$ to $n$, with $n$ being the dimensionality of the manifold $\mathcal{M}$. The metric allows for measuring lengths and angles on the manifold. The length of a smooth curve \(\gamma: [a, b] \rightarrow \mathcal{M}\) is given by:
\begin{equation}
L(\gamma) = \int_a^b \sqrt{g_{ij} \, \frac{d\gamma^i}{d\tau} \frac{d\gamma^j}{d\tau}} \, d\tau.
\end{equation}

A \textit{geodesic} on a manifold is a curve whose velocity vector remains parallel to itself along the curve, representing the straightest possible path given the manifold’s geometry. This property is formally expressed by stating that the geodesic has zero covariant acceleration, which accounts for the curvature of the manifold rather than the usual notion of acceleration in Euclidean space. Mathematically, a curve \(\gamma(\tau)\) is a geodesic if it satisfies the geodesic equation:
\begin{equation}
\frac{d^2 x^i}{d\tau^2} + \Gamma^i_{jk} \frac{dx^j}{d\tau} \frac{dx^k}{d\tau} = 0,
\label{geod_equ}
\end{equation}
where \(\tau\) is an affine parameter along the curve, and \(\Gamma^i_{jk}\) are the Christoffel symbols of the second kind, defined by:
\begin{equation}
\Gamma^i_{jk} = \frac{1}{2} g^{il} \left( \partial_j g_{lk} + \partial_k g_{lj} - \partial_l g_{jk} \right).
\label{DefChristoffel}
\end{equation}

The Christoffel symbols define the covariant derivative, which maps tensor fields to other tensor fields, adapting to the curvature of the manifold. Specifically, the covariant derivative of a vector field \(V^j\) is defined as:
\begin{equation}
\nabla_i V^j = \partial_i V^j + \Gamma^j_{ik} V^k,
\end{equation}
This expression generalizes the notion of directional derivatives to curved manifolds and ensures invariance under changes of coordinates.

The \textit{exponential map} is a tool in Riemannian geometry that relates the tangent space at a point to the manifold. Given a tangent vector \(v \in T_p\mathcal{M}\), the exponential map \(\text{Exp}_p
\) maps \(v\) to a point on the manifold reached by traveling along the geodesic starting at \(p\) with initial velocity \(v\). Formally,
\begin{equation}
\text{Exp}_p(v) = \gamma(1),
\end{equation}
where \(\gamma(\tau)\) is the geodesic satisfying the initial conditions \(\gamma(0) = p\) and \(\dot{\gamma}(0) = \left.\frac{d\gamma(\tau)}{d\tau}\right|_{\tau=0} = v\). By rescaling the parameter \(v\) by a factor of \(\epsilon\), the following relation holds:
\begin{equation}
\text{Exp}_p(\epsilon v) = \gamma(\epsilon).
\label{exp_geo}
\end{equation}

To approximate \(\text{Exp}_p(\epsilon v)\) locally, we expand \(\gamma(\epsilon)\) for small \(\epsilon\) using a Taylor series:
\begin{equation}
\text{Exp}_p(\epsilon v) \approx \gamma(0) + \epsilon \dot{\gamma}(0) + \frac{\epsilon^2}{2} \ddot{\gamma}(0) +  \mathcal{O}(\epsilon^3).
\end{equation}

This expansion provides insight into the behavior of the geodesic near the starting point \( p \). The first-order term, \(\gamma(0) + \epsilon \dot{\gamma}(0)\), corresponds to the initial position and velocity, providing a linear approximation of the geodesic. The second-order term, \(\frac{\epsilon^2}{2}\ddot{\gamma}(0)\), accounts for curvature effects, refining the approximation by incorporating the geodesic's acceleration. Since \(\ddot{\gamma}(0)\) is given by evaluating the geodesic equation \eqref{geod_equ} at \(\tau = 0\), it encodes the local curvature effects at the starting point.

In the context of quantum information, particularly for the QNG, the manifold of interest is the parameter space of quantum states, where the metric tensor relevant for us is defined by the Fubini-Study metric. This metric provides a Riemannian structure that captures the infinitesimal distance between quantum states, facilitating more efficient optimization in variational quantum algorithms. In the next sections, we review how to apply the Fubini-Study metric to the Quantum Natural Gradient and how to distinguish geodesic corrections to QNG from the perspective of the exponential map.


\section{Optimizing Idealized Variational Quantum Circuits with Quantum Natural Gradient}\label{Circuits}

This section provides an overview of VQAs and their optimization, focusing on idealized variational quantum circuits. For a more comprehensive understanding, refer to \cite{Cerezo2021, McClean2016, Bharti2022} for reviews of VQAs and \cite{Stokes2020} for details on QNG.

Variational quantum circuits are constructed using a family of parameterized unitary transformations. For an \( n \)-qubit system, the state space is represented by a Hilbert space of dimension \( N = 2^n \), which can be decomposed as a tensor product of two-dimensional spaces: \( \mathbb{C}^N = (\mathbb{C}^2)^{\otimes n} \). The parameterized circuits are typically composed of sequences of unitary transformations:
\begin{equation}
U(\bm{\theta}) = V_l U_l(\theta_l) \cdots V_2 U_2(\theta_2) V_1 U_1(\theta_1),
\end{equation}
where \( V_j \) are fixed unitary operators, and \( U_j(\theta_j) \) are parameterized gates. Here, \(\bm{\theta} = (\theta_1,\dots,\theta_l)\) is the vector of parameters, and the index \( l \) denotes the total number of parameterized gates in the circuit. Each \( U_j(\theta_j) \) is defined as:
\begin{equation}
U_j(\theta_j) = e^{-i \frac{\theta_j}{2} K_j},
\end{equation}
with \( K_j \) being Hermitian operators.

The goal of VQE is to minimize a cost function, typically defined as the expectation value of an observable \( \hat{O} \) with respect to the quantum state \( \ket{\psi(\bm{\theta})} = U(\bm{\theta}) \ket{\psi_0} \):
\begin{align}
\mathcal{L}(\bm{\theta}) &= \bra{\psi(\bm{\theta})} \hat{O} \ket{\psi(\bm{\theta})},
\end{align}
where \( \ket{\psi_0} \) is the initial state. The objective is to find the optimal parameters \( \bm{\theta}^* \) that minimize \( \mathcal{L}(\bm{\theta}) \).

The optimization is typically performed using gradient descent, updating the parameters iteratively:
\begin{align}
\bm{\theta}_{t+1} &= \bm{\theta}_t - \eta \, \boldsymbol{\nabla} \mathcal{L}(\bm{\theta}_t),
\label{gradient_descent}
\end{align}
where \( \boldsymbol{\nabla}:= (\partial_1, \dots, \partial_l) = \left( \frac{\partial}{\partial \theta^1}, \dots, \frac{\partial}{\partial \theta^l} \right) \) denotes the gradient operator with respect to the parameter vector \( \bm{\theta} = (\theta_1, \dots, \theta_l) \).

Equation~\eqref{gradient_descent} can be interpreted as an approximation of the solution to an ordinary differential equation (ODE) using the Euler method:
\begin{align}
\dot{\bm{\theta}} &= - \lambda \, \boldsymbol{\nabla} \mathcal{L}(\bm{\theta}),
\end{align}
where \( \eta = \epsilon \lambda \) is the learning rate, with \( \lambda \) being a time scale constant that affects the speed but not the trajectory of the system, and \( \epsilon \) is the step size. 

However, this ODE is not invariant under reparameterizations of the parameters \( \bm{\theta} \). For instance, if we rescale the parameters \( \bm{\theta} \rightarrow 2\bm{\theta} \), the gradient \( \boldsymbol{\nabla} \mathcal{L} \) would scale as \( \frac{1}{2} \boldsymbol{\nabla} \mathcal{L} \), leading to inconsistencies in the optimization process.

The core of this issue lies in the differential geometric nature of the gradient. The parameter update \( \dot{\bm{\theta}} \) transforms as a vector in the tangent space \( T_{\bm{\theta}} \mathcal{M} \), while the gradient \( \boldsymbol{\nabla} \mathcal{L} \) is a covector (or 1-form) in the cotangent space \( T^*_{\bm{\theta}} \mathcal{M} \). Since the ODE in Eq. \eqref{gradient_descent} attempts to relate objects in different spaces with distinct transformation rules, it is not an invariant relation.

QNG alleviates this issue by approximately solving an invariant ODE. The key idea is to apply the inverse metric \( \mathbf{g}^{-1} \) to the gradient vector, yielding the natural gradient \( \mathbf{g}^{-1} \boldsymbol{\nabla} \mathcal{L}(\boldsymbol{\theta}) \). The resulting ODE becomes:
\begin{equation}
\dot{\boldsymbol{\theta}} = -\lambda\, \mathbf{g}^{-1} \boldsymbol{\nabla} \mathcal{L}(\boldsymbol{\theta}),
\label{QNG_ODE}
\end{equation}
which is now a vector in \( T_{\bm{\theta}} \mathcal{M} \), thereby resolving the type mismatch in Eq. \eqref{gradient_descent}.

This new ODE is invariant under reparameterizations, ensuring that the forward Euler approximation:
\begin{align}
\bm{\theta}_{t+1} &= \bm{\theta}_t - \eta \mathbf{g}^{-1} \boldsymbol{\nabla} \mathcal{L}(\bm{\theta}_t),
\end{align}
remains consistent across different parameter spaces. The metric tensor \( g_{ij} \) is recognized as the Fubini-Study metric, derived from the real part of the Quantum Geometric Tensor \cite{Provost1980}:
\begin{align}
g_{i j} &= \text{Re} \left( \left\langle \partial_i \psi | \partial_j \psi \right\rangle \right) - \left\langle \partial_i \psi | \psi \right\rangle \left\langle \psi | \partial_j \psi \right\rangle.
\label{Fubini}
\end{align}

It is important to note that the metric tensor in VQAs does not always define a legitimate Riemannian metric, as it is often degenerate, meaning it may not be invertible. To address this issue, regularization techniques such as Tikhonov regularization are applied by adding a small constant multiplied by the identity matrix (\(\lambda I\)) to the metric, ensuring it is well-defined and invertible. Additionally, the computation of this metric tensor becomes increasingly intensive as the number of parameters in the circuit grows. Therefore, approximations, such as block-diagonal or diagonal forms, are crucial for practical applications \cite{Stokes2020}.

Thus, QNG utilizes a regularized approximated, often diagonal or block-diagonal, inverse metric tensor to perform parameter updates that respect the intrinsic geometry of the quantum state space, leading to improved convergence speed and accuracy in optimization within VQAs. In the next section, we explore how this approach can be further refined through higher-order integrators and the inclusion of geodesic corrections.
\section{Higher-order Integrators and Geodesic Correction}
\label{Higher}

The forward Euler method, commonly used in the QNG, provides only a first-order approximation to the exact solution of the natural gradient ordinary differential equation (ODE). For higher accuracy, higher-order integrators should be employed.

To further refine this approach, we can use the Riemannian Euler method, which leverages the exponential map to update the parameters, ensuring that the updates align with the geometry of the manifold. Using \eqref{exp_geo}, the Riemannian Euler update rule is given by:
\begin{align}
\boldsymbol{\gamma}_t(\epsilon) &= \boldsymbol{\theta}_{t+1} = \text{Exp}_{\boldsymbol{\theta}_t}\left(-\epsilon \lambda \, \mathbf{g}^{-1}(\boldsymbol{\theta}_t) \, \boldsymbol{\nabla} \mathcal{L}(\boldsymbol{\theta}_t)\right),
\label{riemann_update}
\end{align}
where
\begin{align}
\boldsymbol{\gamma}_t(0) &= \boldsymbol{\theta}_t, \\
\dot{\boldsymbol{\gamma}}_t(0) &= -\lambda \, \mathbf{g}^{-1}(\boldsymbol{\theta}_t) \, \boldsymbol{\nabla} \mathcal{L}(\boldsymbol{\theta}_t).
\end{align}

Here, the exponential map \( \text{Exp} \) translates the current parameter \(\bm{\theta}_t\) along the geodesic defined by the natural gradient. This approach is effective because it preserves invariance properties under reparameterization, owing to the characteristics of the exponential map. However, directly computing the exponential map is challenging, as it requires solving the geodesic equation exactly. To approximate this computation, we explore using the first and second-order derivatives as detailed in Section \ref{Differential}.

The first derivatives approximate the geodesic as:
\begin{equation}
\boldsymbol{\gamma}_t(\epsilon) \approx \boldsymbol{\gamma}_t(0) + \epsilon \dot{\boldsymbol{\gamma}}_t(0) 
\Rightarrow \boldsymbol{\theta}_{t+1} = \boldsymbol{\theta}_t - \eta \, \mathbf{g}^{-1} \boldsymbol{\nabla} \mathcal{L}(\boldsymbol{\theta}_t),
\end{equation}
where the learning rate \(\eta = \epsilon \lambda\). This approximation corresponds to the naive Quantum Natural Gradient update rule, utilizing only the first-order information. For a more precise approximation, we can incorporate second-order information from the geodesic equation, we have:

\begin{equation}
\boldsymbol{\gamma}_t(\epsilon) \approx \boldsymbol{\gamma}_t(0) + \epsilon \dot{\boldsymbol{\gamma}}_t(0) + \frac{1}{2} \epsilon^2 \ddot{\boldsymbol{\gamma}}_t(0).
\end{equation}
Evaluating the geodesic equation \eqref{geod_equ} at $\tau = 0$ gives:

\begin{equation}
\ddot{\boldsymbol{\gamma}}_t(0) = - \, \boldsymbol{\Gamma}(\boldsymbol{\theta}_t)\left[\dot{\boldsymbol{\gamma}}_t(0), \dot{\boldsymbol{\gamma}}_t(0)\right].
\end{equation}
Here, \( \boldsymbol{\Gamma}[\cdot, \cdot] \) denotes the action of the connection (Christoffel symbols) as a bilinear map on tangent vectors; in components, it corresponds to \( \Gamma^i_{lm} \dot{\gamma}^l \dot{\gamma}^m \).

The resulting update rule with the geodesic correction is:

\begin{equation}
\boldsymbol{\theta}_{t+1} = \boldsymbol{\theta}_t + \epsilon \dot{\boldsymbol{\gamma}}_t(0) - \frac{1}{2} \epsilon^2 \, \boldsymbol{\Gamma}(\boldsymbol{\theta}_t)\left[\dot{\boldsymbol{\gamma}}_t(0), \dot{\boldsymbol{\gamma}}_t(0)\right].
\label{geo_update}
\end{equation}

We now combine all the relevant equations and heuristically allow the correction term in equation \eqref{geo_update} to depend on a tunable parameter \( b \) rather than being fixed to \(\eta^2\). In other words, \(\eta^2\) is too small to effectively capture the geodesic correction effect. This approach is justified because we use an approximate Fubini-Study metric, such as a diagonal metric, and fully capturing the curvature may require going beyond a second-order approximation in \eqref{riemann_update}. Such higher-order corrections would be computationally intensive. To avoid this complexity, we adjust the correction term flexibly with the parameter \( b \), allowing us to capture the geodesic correction effects without the need for exact higher-order terms, which are computationally demanding.

This leads to the following update rule, incorporating the geodesic correction into the QNG:

\begin{equation}
\boldsymbol{\theta}_{t+1} = \boldsymbol{\theta}_t 
- \eta \, \mathbf{g}^{-1}(\boldsymbol{\theta}_t) \boldsymbol{\nabla} \mathcal{L}(\boldsymbol{\theta}_t)
- \frac{b}{2} \, \boldsymbol{\Gamma}(\boldsymbol{\theta}_t) 
\left[ \mathbf{g}^{-1} \boldsymbol{\nabla} \mathcal{L}(\boldsymbol{\theta}_t), \, \mathbf{g}^{-1} \boldsymbol{\nabla} \mathcal{L}(\boldsymbol{\theta}_t) \right].
\label{geo_update_b}
\end{equation}
This is equivalent, in components, to the expression:
\begin{equation}
\theta^{\mu}_{t+1} = \theta^{\mu}_t 
- \eta \, g^{ij}(\theta^{\mu}_t) \, \partial_j \mathcal{L}(\theta^{\mu}_t)
- \frac{b}{2} \, \Gamma^i_{lm}(\theta^{\mu}_t) 
\left( g^{lj}(\theta^{\mu}_t) \, \partial_j \mathcal{L}(\theta^{\mu}_t) \right)
\left( g^{mj}(\theta^{\mu}_t) \, \partial_j \mathcal{L}(\theta^{\mu}_t) \right).
\end{equation}
where \( \theta^{\mu} \) denotes the components of the parameter vector \( \boldsymbol{\theta} \) in a local coordinate, and \( \partial_j := \frac{\partial}{\partial \theta^j} \).

The update rule in equation \eqref{geo_update_b} shows that the correction term depends on the Christoffel symbols of the second kind. In the next section, we will discuss how to compute these symbols efficiently.

\section{Computing the Christoffel Symbols with the Parameter-Shift Rule}\label{Christoffel}

The Christoffel symbols are essential components in the update rule \eqref{geo_update_b}, as they capture the curvature of the parameter space defined by the Fubini-Study metric. Traditionally, the Christoffel symbols are computed by differentiating the Fubini-Study metric \eqref{Fubini}, which involves first calculating the metric tensor and then deriving the Christoffel symbols through classical differentiation. This approach relies heavily on classical post-processing and does not fully utilize information available from parameterized quantum states.

To enable direct estimation of the Christoffel symbols, we reformulate their computation using the parameter-shift rule specifically adapted to parameterized quantum circuits. This method allows the Christoffel symbols to be directly estimated from measurement outcomes, bypassing the need for classical differentiation and making the process more efficient and directly linked to quantum state preparation. For a detailed overview of the parameter-shift rule and its application in deriving the Fubini-Study metric, please refer to the appendix and the related work in \cite{Mari}.

The Fubini-Study metric for a pure variational quantum state \( |\psi(\bm{\theta}) \rangle \) can be represented as a second-order tensor:

\begin{equation}
g_{j_1 j_2} (\bm{\theta}) = -\frac{1}{2} \frac{\partial^2}{\partial \theta_{j_1} \partial \theta_{j_2}} \left. \left| \langle \psi (\bm{\theta}') | \psi (\bm{\theta}) \rangle \right|^2 \right|_{\boldsymbol{\theta}' = \boldsymbol{\theta}}
\end{equation}
To evaluate this second derivative, we apply a parameter-shift rule~\cite{Mitarai,Schuld} adapted for fidelity overlaps, following the formulation introduced in~\cite{Mari}. This method uses central finite differences with shifts of \( \pm \pi/2 \), yielding an exact expression for the metric from four circuit evaluations.

\begin{align}
g_{j_1 j_2} (\bm{\theta}) = -\frac{1}{8} &\Bigg[ 
\left| \langle \psi (\bm{\theta}) | \psi (\bm{\theta} + (\mathbf{e}_{j_1} + \mathbf{e}_{j_2}) \pi / 2) \rangle \right|^2 \nonumber \\
&- \left| \langle \psi (\bm{\theta}) | \psi (\bm{\theta} + (\mathbf{e}_{j_1} - \mathbf{e}_{j_2}) \pi / 2) \rangle \right|^2 \nonumber \\
&- \left| \langle \psi (\bm{\theta}) | \psi (\bm{\theta} + (-\mathbf{e}_{j_1} + \mathbf{e}_{j_2}) \pi / 2) \rangle \right|^2 \nonumber \\
&+ \left| \langle \psi (\bm{\theta}) | \psi (\bm{\theta} - (\mathbf{e}_{j_1} + \mathbf{e}_{j_2}) \pi / 2) \rangle \right|^2 
\Bigg].
\label{full_metric}
\end{align}
where \( \mathbf{e}_{j_1} \) and \( \mathbf{e}_{j_2} \) are unit vectors in the parameter space. If the metric tensor is assumed to be diagonal (i.e., \( g_{j_1 j_2} = 0 \) for \( j_1 \neq j_2 \)), the expression simplifies considerably and reduces the number of required evaluations:
\begin{equation}
g_{j j}(\boldsymbol{\theta}) = \frac{1}{4} \left[ 1 - \left| \langle \psi(\boldsymbol{\theta}) | \psi(\boldsymbol{\theta} + \pi e_j) \rangle \right|^2 \right].
\label{diag_metric}
\end{equation}
In this scenario, the Christoffel symbols can be computed directly from the derivatives of the metric tensor using the higher-order parameter-shift rule, enabling their estimation directly from the quantum circuit. For the diagonal case, the results are summarized in the following proposition:
\begin{proposition}
The Christoffel symbols of the second kind for the diagonal metric \eqref{diag_metric} are given by:
\begin{equation}
\begin{aligned}
\Gamma^i_{jk} &= \frac{1}{16g_{ii}(\theta)} \Bigg[   \delta_{ij} 
 \Bigg( 
- \left| \langle \psi(\boldsymbol{\theta}) | \psi(\boldsymbol{\theta} + \pi e_i +\dfrac{\pi}{2} e_k) \rangle \right|^2 
+ \left| \langle \psi(\boldsymbol{\theta}) | \psi(\boldsymbol{\theta} + \pi e_i - \dfrac{\pi}{2} e_k) \rangle \right|^2
\Bigg)\\
&+   \delta_{ik} 
 \Bigg( 
- \left| \langle \psi(\boldsymbol{\theta}) | \psi(\boldsymbol{\theta} + \pi e_i +\dfrac{\pi}{2} e_j) \rangle \right|^2 
+ \left| \langle \psi(\boldsymbol{\theta}) | \psi(\boldsymbol{\theta} + \pi e_i - \dfrac{\pi}{2} e_j) \rangle \right|^2
\Bigg) \\
&-   \delta_{jk} 
 \Bigg( 
- \left| \langle \psi(\boldsymbol{\theta}) | \psi(\boldsymbol{\theta} + \pi e_j +\dfrac{\pi}{2} e_i) \rangle \right|^2 
+ \left| \langle \psi(\boldsymbol{\theta}) | \psi(\boldsymbol{\theta} + \pi e_j - \dfrac{\pi}{2} e_i) \rangle \right|^2
\Bigg) 
\Bigg].
\end{aligned}
\label{Christoffel_parshiff}
\end{equation}
\end{proposition}

\begin{proof}
See Appendix~\ref{app:Christoffel}.
\end{proof}

The computation of the metric tensor and Christoffel symbols involves evaluating combinations of state overlaps at each step of the update rule algorithm \eqref{geo_update_b}. Specifically, the Christoffel symbols \eqref{Christoffel_parshiff} require six distinct expectation values. For instance, one of these overlaps can be expressed as:
\begin{equation}
|\langle \psi(\bm{\theta}) | \psi(\bm{\theta} + \pi e_i + \frac{\pi}{2} e_k) \rangle|^2 
= |\langle 0 | U^{\dagger}(\bm{\theta}) U(\bm{\theta} + \pi e_i + \frac{\pi}{2} e_k) | 0 \rangle|^2.
\end{equation}

This overlap can be intuitively interpreted as the probability of measuring all qubits of the quantum state \( U(\bm{\theta}) U(\bm{\theta} + \pi e_i + \frac{\pi}{2} e_k) | 0 \rangle \) in the computational basis at \(|0\rangle\). If \( l \) denotes the number of variational parameters in the quantum circuit generating the state \(|\psi(\bm{\theta})\rangle\), then computing the full Fubini–Study metric requires \( O(l^2) \) quantum circuits of depth \( 2d \), and the Christoffel symbols require \( O(l^3) \) quantum circuits of the same depth. Here, \( d \) refers to the depth of a single execution of the variational ansatz, i.e., the number of sequential gate layers. For the diagonal approximation of the metric, the cost reduces to \( O(l) \) quantum circuits, while the Christoffel symbols require \( O(l^2) \) quantum circuits, both still with depth \( 2d \).

\section{Application Examples}\label{Examples}

In this section, we present examples to illustrate the practical application of the previously developed theories.

\subsection{Example 1 : Single Qubit}
In this example, we examine a single qubit scenario. For the VQE, we take the ansatz state defined as:
\begin{equation}
|\phi(\bm{\theta})\rangle = \cos\theta_0|0\rangle + e^{2i\theta_1}\sin\theta_0|1\rangle = \begin{bmatrix}
\cos\theta_0 \\
e^{2i\theta_1}\sin\theta_0
\end{bmatrix}.
\label{exp1:state}
\end{equation}
where \( \boldsymbol{\theta} = (\theta_0, \theta_1) \).

The Hamiltonian for which we want to find the ground state is \( H = \sigma_x \), where the energy is expressed in dimensionless units in this pedagogical VQE example. Thus, the cost function for the system is:
\begin{equation}
f = \langle \phi(\theta)|H|\phi(\theta)\rangle = \sin(2\theta_0)\cos(2\theta_1),
\end{equation}
and its gradient vector is:
\begin{equation}
\frac{\partial f}{\partial \bm{\theta}} = \begin{bmatrix}
2\cos(2\theta_0)\cos(2\theta_1) \\
-2\sin(2\theta_0)\sin(2\theta_1)
\end{bmatrix}.
\end{equation}

The Fubini-Study metric for this single qubit example, using equation \eqref{Fubini}, is given by:
\begin{equation}
F = \begin{bmatrix}
1 & 0 \\
0 & \sin^2(2\theta_0)
\end{bmatrix}.
\end{equation}

The Christoffel symbols of this metric, calculated using Eq.~\eqref{DefChristoffel}, have the following nonzero components:

\begin{equation}
\begin{aligned}
\Gamma^{0}_{\; 1 1} &= - 2 \sin(2 \theta_0) \cos(2 \theta_0),  \\
\Gamma^{1}_{\; 0 1} &=\Gamma^{1}_{\; 1 0} =   2 \dfrac{\cos(2 \theta_0)}{\sin(2 \theta_0)}
\end{aligned}
\end{equation}

\begin{figure}[H]
    \centering
    \includegraphics[width=\textwidth]{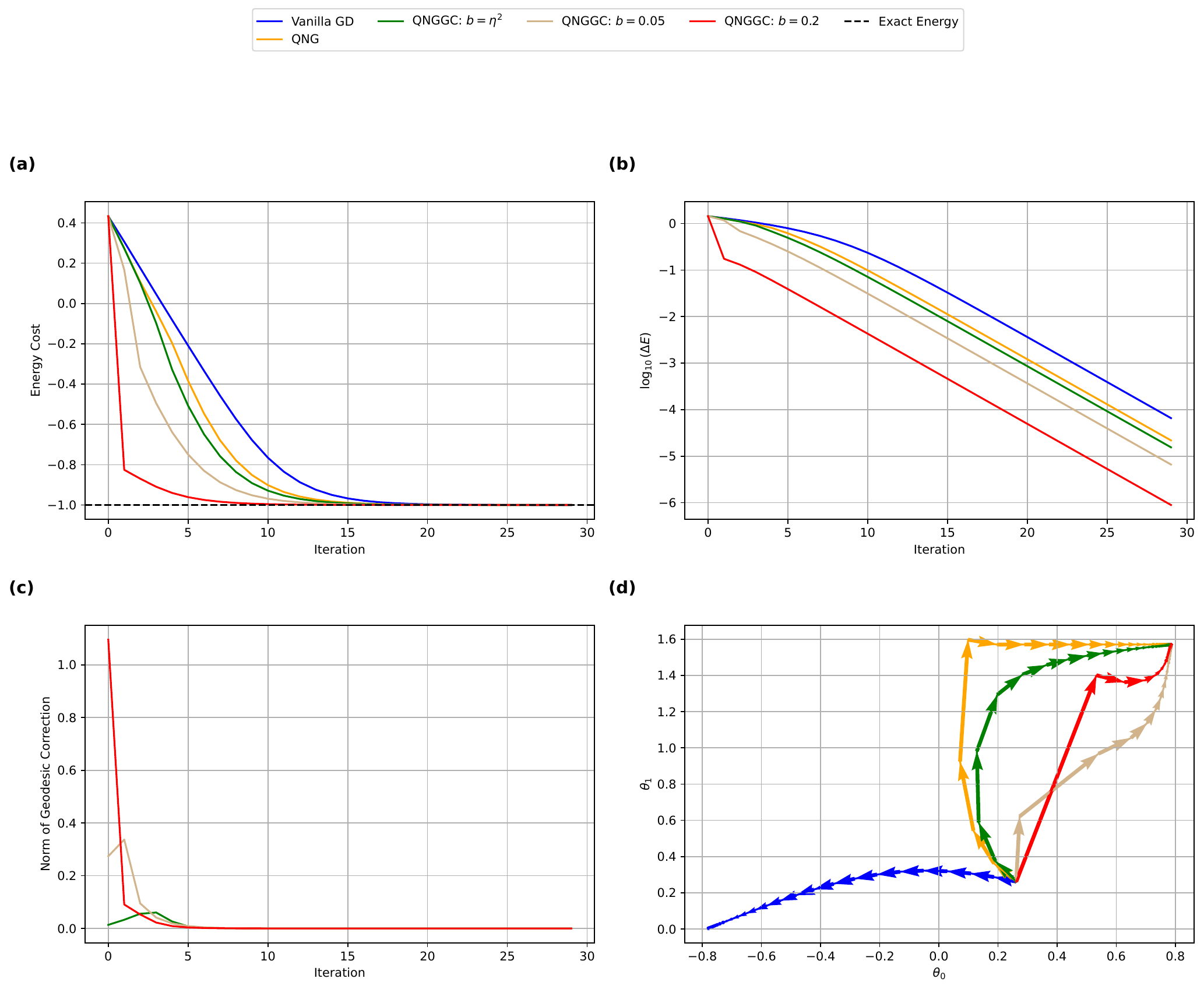}
    \caption{
    Performance of GD, QNG, and QNGGC with varying correction parameters (\( b = \eta^2 \), \( b = 0.05 \), and \( b = 0.2 \)). The initial parameters \((\theta_0, \theta_1)\) are set to \([\pi/12, \pi/12]\), corresponding to a specific initial quantum state defined in Eq.~\eqref{exp1:state}. A learning rate \(\eta = 0.05\) is used for all optimizers. 
    (a) Shows the evolution of the energy cost function over 30 iterations, demonstrating the improved convergence behavior of QNGGC relative to QNG and GD.
    (b) Depicts the logarithmic energy difference, \(\log_{10}(\Delta E)\), where \(\Delta E\) represents the gap between the current energy and the target energy \( E^* = -1 \). The current energy is computed as the expectation value \(\langle \phi(\theta) | H | \phi(\theta) \rangle\) at each step. 
    (c) Displays the norm of the geodesic correction term across different correction parameters \(b\), highlighting the impact of geodesic corrections during optimization.
    (d) Shows the trajectory of the parameters \(\theta_0\) and \(\theta_1\) throughout the optimization process, illustrating how each method navigates the parameter space. The QNGGC paths are more direct, emphasizing the benefits of incorporating geodesic corrections.
    }
    \label{fig:Ex1_E_plots}
\end{figure}

Figure \ref{fig:Ex1_E_plots} demonstrates the impact of geodesic corrections on optimization performance. The inclusion of geodesic corrections in QNGGC significantly enhances convergence speed (in terms of steps) compared to GD and QNG, as depicted in subplots (a) and (b), where QNGGC rapidly approaches the target energy \( E^* = -1 \). Subplot (c) shows the norm of the geodesic correction term, $\frac{b}{2} \Gamma^i_{lm}(\theta^{\mu}_t) 
\left( g^{lj}(\theta^{\mu}_t) \partial_j \mathcal{L}(\theta^{\mu}_t) \right)
\left( g^{mj}(\theta^{\mu}_t) \partial_j \mathcal{L}(\theta^{\mu}_t) \right)$, which plays a crucial role in guiding the optimization process, particularly in the initial iterations. The subsequent reduction of this norm towards zero indicates that the convergence of the gradient is being effectively managed, leading to more stable updates. Subplot (d) illustrates the optimization paths in the parameter space, where QNGGC provides a more direct and efficient trajectory compared to GD and QNG. This behavior underscores the advantage of leveraging curvature corrections, resulting in more precise and adaptive parameter updates.
\subsection{Example 2: Two-Qubit Simulation of the Hydrogen Molecule (\(\text{H}_2\))}

In this example, we focus on finding the ground state of the hydrogen molecule (\(\text{H}_2\)) using a two-qubit variational quantum eigensolver (VQE) approach. The Hamiltonian for the system is given by \cite{Yamamoto}:
\begin{equation}
H = \alpha (\sigma_z \otimes I + I \otimes \sigma_z) + \beta \sigma_x \otimes \sigma_x,
\end{equation}

where \(\alpha = 0.4\) and \(\beta = 0.2\). In this Hamiltonian, \(\sigma_z\) and \(\sigma_x\) are Pauli matrices that represent spin operators acting on individual qubits. The energy is expressed in Hartree atomic units, which is standard for molecular electronic structure simulations. This Hamiltonian has four eigenvalues:
\begin{equation}
h_1 = \sqrt{4\alpha^2 + \beta^2}, \quad h_2 = \beta, \quad h_3 = -\beta, \quad h_4 = -\sqrt{4\alpha^2 + \beta^2},
\end{equation}
with the ground state corresponding to the minimum eigenvalue \(h_4\). The corresponding eigenvector is the ground state, \(\ket{\psi_{\text{min}}}\), which is expressed as:
\begin{equation}
\ket{\psi_{\text{min}}} \propto -\beta \ket{00} + (2\alpha + \sqrt{4\alpha^2 + \beta^2}) \ket{11}.
\end{equation}

The ansatz for the VQE is designed to approximate this ground state. It is given as follows (refer to Figure \ref{fig_exp2} for a visual representation):
\begin{equation}
\begin{aligned}
|\psi\rangle = (CRY(\theta)_{q_0, q_1}) (CRX(\theta)_{q_0, q_1}) \left( R_y(2\theta_0) \otimes R_y(2\theta_1) \right) \ket{0} \otimes \ket{0}\\
\end{aligned}
\label{ansatz_ex2}
\end{equation}
where \( R_y(\theta) \) denotes the single-qubit rotation operator defined as \( R_y(\theta) = e^{-i\theta\sigma_y/2} \), and the entangling gates act between two qubits, \(q_0\) and \(q_1\), as \( CRX(\theta)_{q_0, q_1} = I \otimes |0\rangle\langle 0| + X \otimes |1\rangle\langle 1| \) and \( CRY(\theta)_{q_0, q_1} = I \otimes |0\rangle\langle 0| + R_y(\theta) \otimes |1\rangle\langle 1| \). Please refer to Appendix \ref{appendix:metric-christoffel} for details on the cost function and its gradient.

\begin{figure}[H]
\centering
\includegraphics[width=0.5\textwidth]{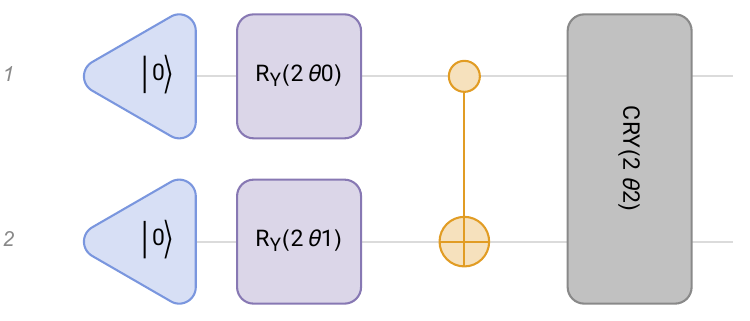}
\caption{An example of the ansatz described in \eqref{ansatz_ex2}.}
\label{fig_exp2}
\end{figure}

In this example, we consider a variational ansatz with three parameters \( \boldsymbol{\theta} = (\theta_0, \theta_1, \theta_2) \), resulting in a \( 3 \times 3 \) Fubini-Study metric \( F \), which is calculated as:

\begin{equation}
F = \begin{bmatrix}
1 & 0 & \cos(\theta_1) \sin(\theta_1) \\
0 & 1 & -\cos(\theta_0) \sin(\theta_0) \\
\cos(\theta_1) \sin(\theta_1) & -\cos(\theta_0) \sin(\theta_0) & \frac{1}{2} \left(1 - \cos(2\theta_0) \cos(2\theta_1)\right)
\end{bmatrix}
\end{equation}

To avoid singularities in \( F \) when updating the parameters in the VQE, we apply Tikhonov regularization by adding a small constant \(\lambda\) as a multiple of the identity matrix \( I \). The inverse Fubini-Study metric \( F^{-1} \) and the Christoffel symbols for this example are provided in Appendix \ref{appendix:metric-christoffel}.

\begin{figure}[H]
\centering
\includegraphics[width=\textwidth]{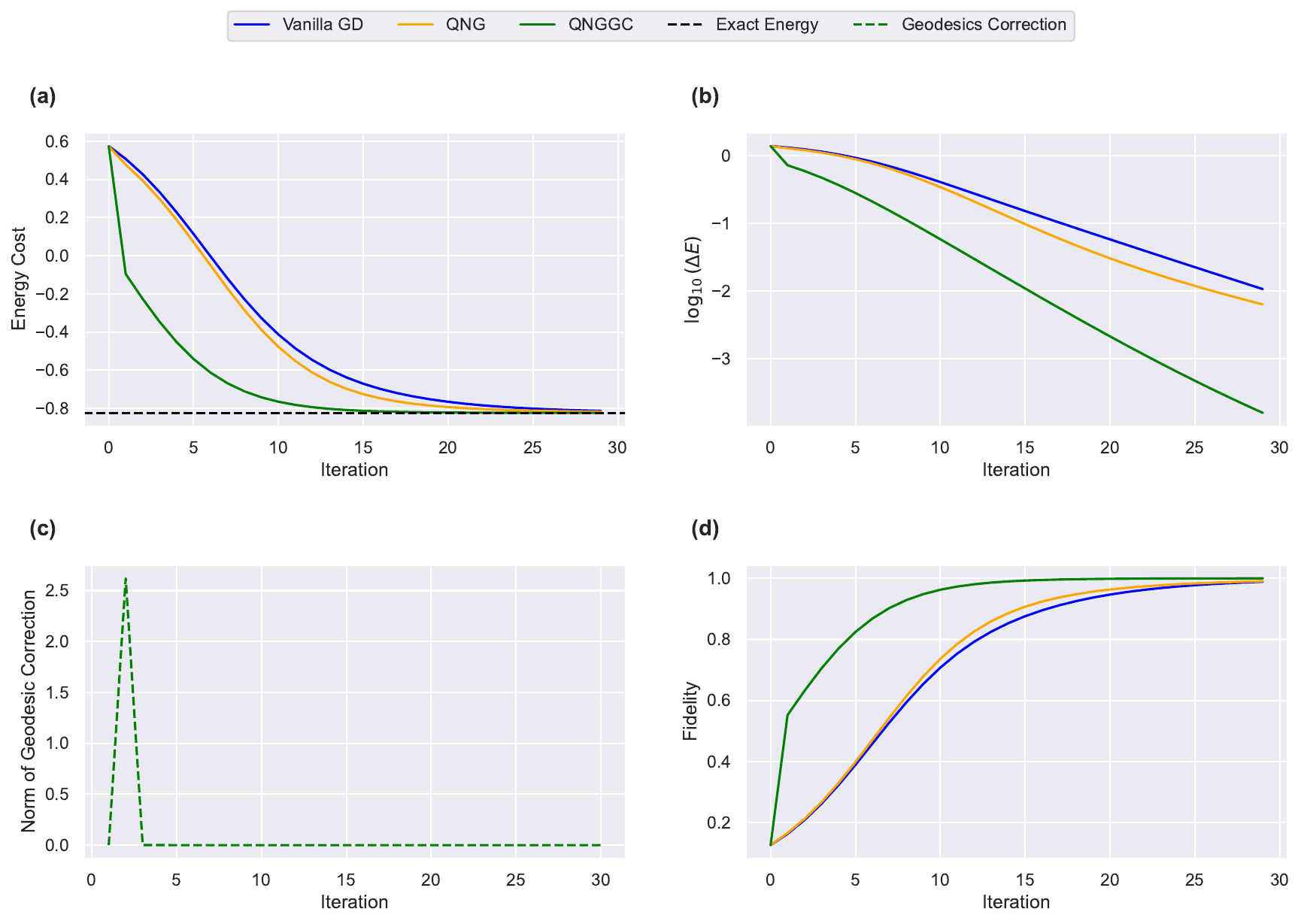}
\caption{
Comparison of the convergence of GD, QNG, and QNGGC. The initial parameters \((\theta_0, \theta_1, \theta_2)\) are set to \([-0.2, -0.2, 0]\), corresponding to the initial state defined by the ansatz in Eq.~\eqref{ansatz_ex2}. The learning rate is \(\eta = 0.05\) for all optimizers, and \(b = \eta^2\) for QNGGC.  
(a) Shows energy convergence.  
(b) Plots \(\log_{10}(\Delta E)\), where \(\Delta E\) is the energy difference from the target energy \(E^* \approx -0.82462\) Hartree.  
(c) Displays the norm of the geodesic correction term in QNGGC.  
(d) Illustrates the fidelity convergence rate. The optimization runs for 30 iterations with a regularization term \(\lambda = 10^{-6}\).
}
\label{fig:exm2_beta0.2_energy}
\end{figure}

Figure~\ref{fig:exm2_beta0.2_energy} demonstrates that QNGGC significantly outperforms GD and QNG, rapidly reducing the energy cost and achieving faster convergence to the target energy. The geodesic correction helps QNGGC navigate the parameter space more effectively, avoiding inefficient trajectories and reaching optimal fidelity faster than the other methods.

\begin{figure}[H]
\centering
\includegraphics[width=\textwidth]{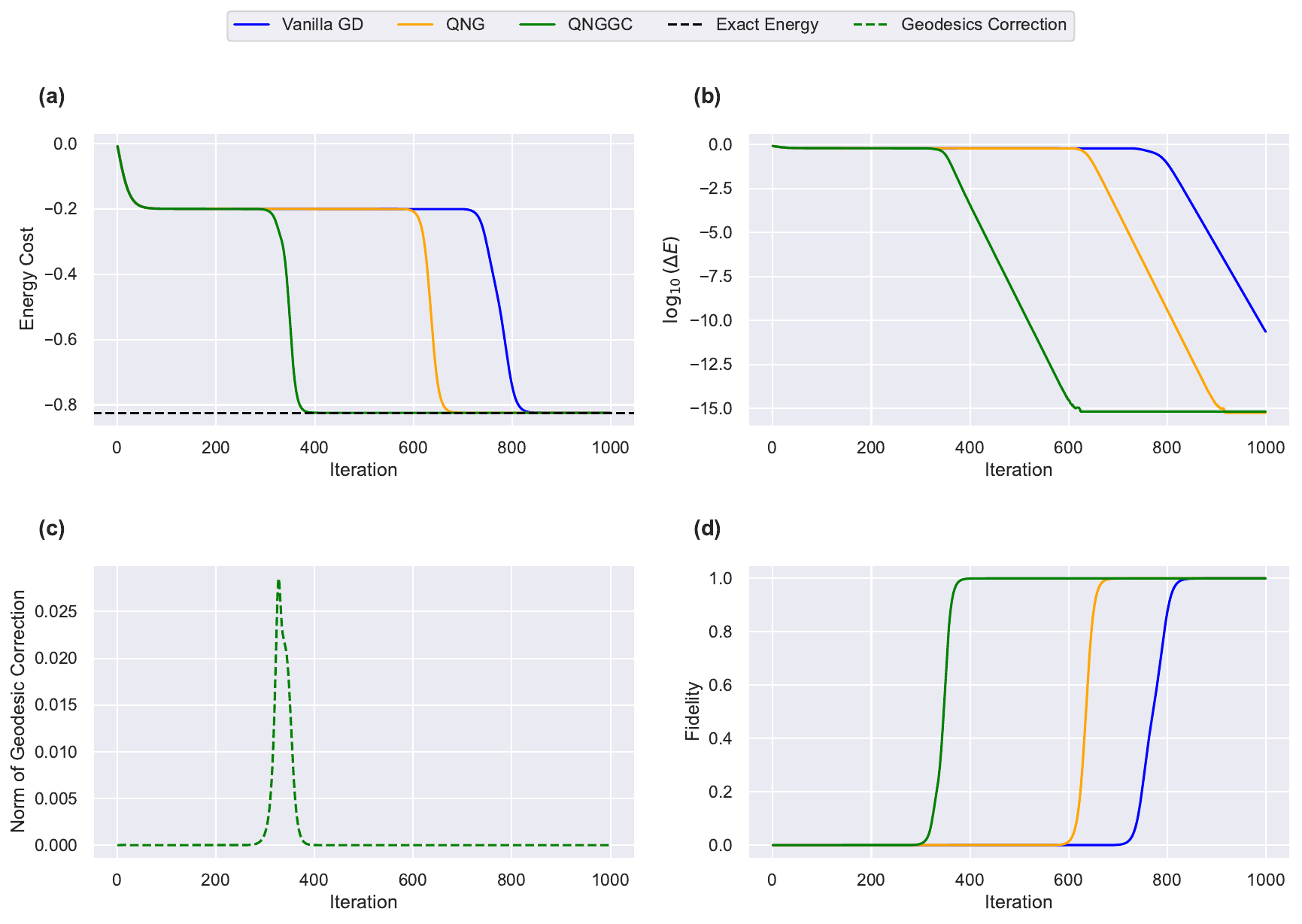}
\caption{
Same conditions as in Figure~\ref{fig:exm2_beta0.2_energy}, except the initial parameters \((\theta_0, \theta_1, \theta_2)\) are set to \([\pi/2, \pi/2, 0]\), and \(b = 0.4\). This configuration is more challenging due to the proximity of the initial state to a flat region of the cost landscape (plateau), which slows down optimization.
}
\label{fig:exm2_beta0.2_plateau}
\end{figure}

Figure~\ref{fig:exm2_beta0.2_plateau} highlights the performance of the optimizers when initialized near a plateau in the cost landscape. Despite the flat region, QNGGC efficiently escapes and converges faster to the ground state compared to GD and QNG.

Figure~\ref{fig:exm2_beta0.2_50} presents averaged results over 50 runs with random initializations, highlighting QNGGC’s robust performance across varied starting conditions.

\begin{figure}[H]
\centering
\includegraphics[width=\textwidth]{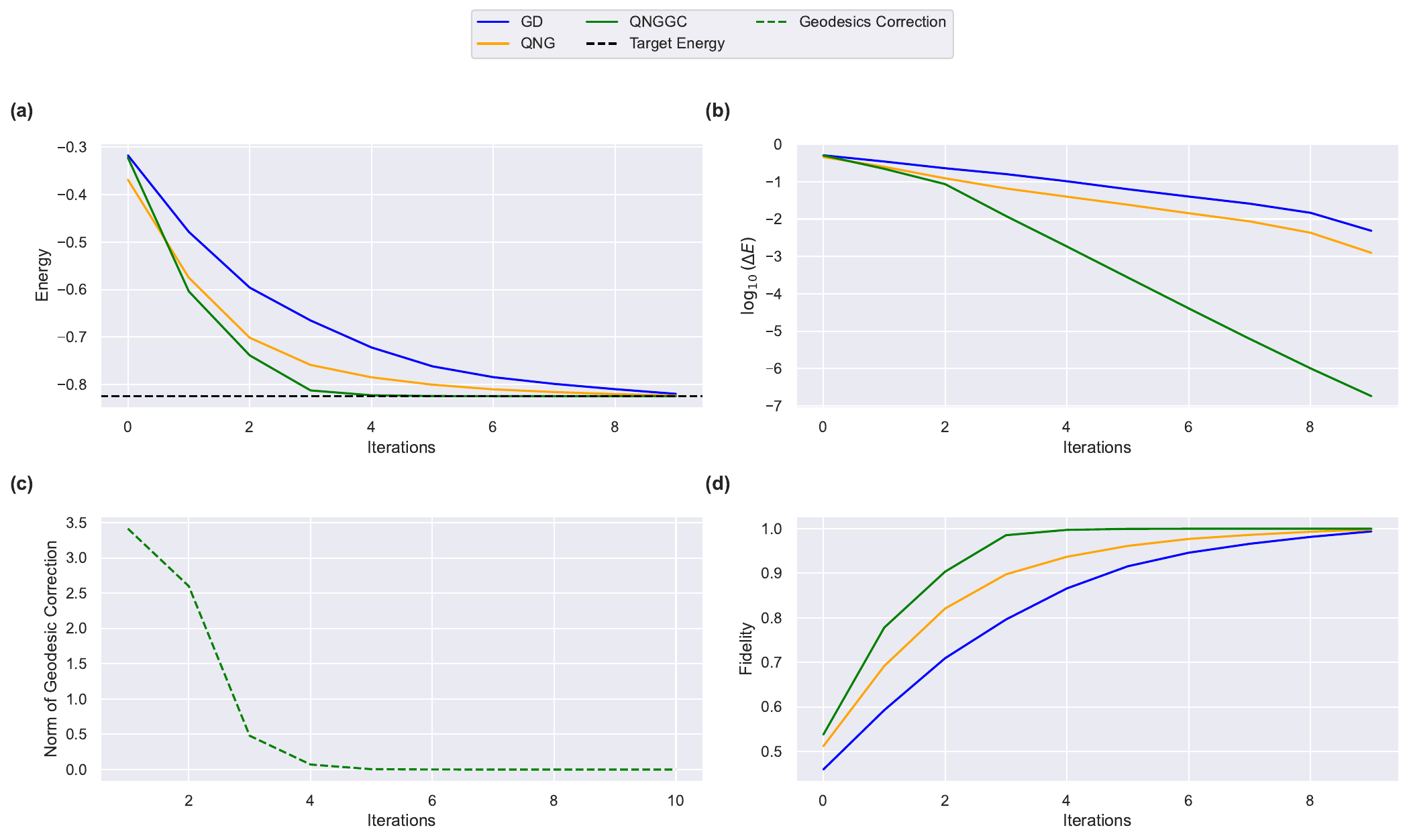}
\caption{
Convergence of GD, QNG, and QNGGC averaged over 50 runs with random initial points. Learning rates and \( b \) were tuned using grid search. The left plot shows energy convergence, while the right plot displays \(\log_{10}(\Delta E)\). All optimizers were run for 10 iterations, with \(\lambda\) set to \(10^{-2}\).
}
\label{fig:exm2_beta0.2_50}
\end{figure}

\subsection{Example 3: Transverse Field Ising Model}

To enhance our previous results, we extend our approach to simulations involving larger qubit systems, including 4 and 7 qubits, using the quantum software Qiskit \cite{Qiskit}. We employ the parameter-shift rule to compute both the gradient and Christoffel symbols. The ansatz used in VQE is the EfficientSU2 ansatz \cite{Kandala2017}, with one repetition as the circuit depth and linear entanglement for simplicity, as illustrated in Figure \ref{fig:EfficientSU2}. Our objective is to determine the ground state energy of the Transverse Field Ising Model \cite{BravoPrieto2020} under open boundary conditions:

\begin{align}
H = -\sum_{i} Z_{i} Z_{i+1} - h \sum_{i} X_{i},
\end{align}

\begin{figure}[H]
\centering
\includegraphics[width=0.5\textwidth]{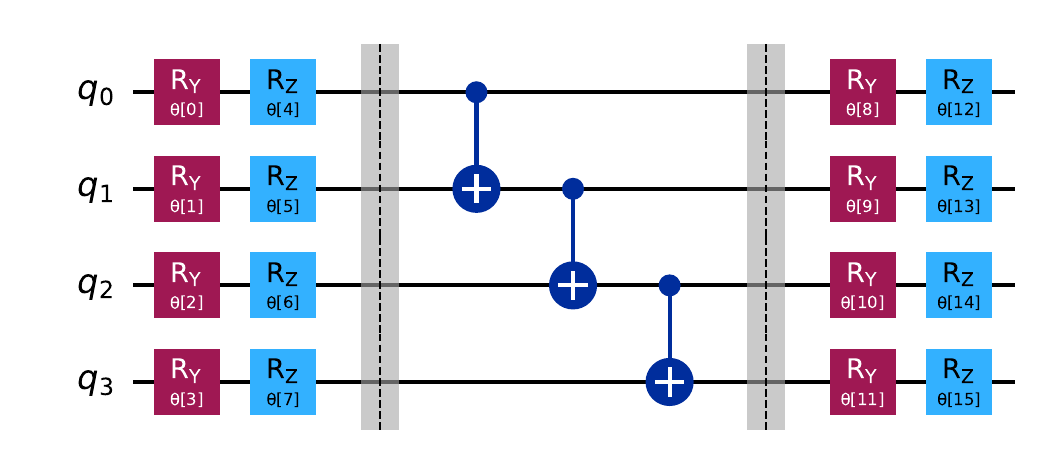}
\caption{EfficientSU2 ansatz for 4 qubits with one repetition circuit depth and linear entanglement.}
\label{fig:EfficientSU2}
\end{figure}

where \( Z_i \) and \( X_i \) are Pauli operators acting on the \(i\)-th qubit, with the transverse field strength denoted by \( h \) and fixed at \( h = 10 \) for our simulations. All simulations are performed in dimensionless units by setting the spin coupling strength to \( J = 1 \); hence, both the energy and transverse field \( h \) are expressed in units of \( J \). This model captures the interplay between nearest-neighbor spin interactions and a transverse magnetic field, driving the system between distinct quantum phases. For smaller values of the transverse field \( h \), the spins tend to align along the z-direction, corresponding to a ferromagnetic ordered phase. As \( h \) increases, the system becomes increasingly disordered and transitions into a paramagnetic phase. The critical point at \( h = 1 \) marks the quantum phase transition between these phases, where spontaneous symmetry breaking and critical quantum fluctuations occur.

To benchmark the effectiveness of our optimizers, we present in Fig. \ref{fig:exm3_4Q}, similar to previous examples, \(\log_{10}(|\langle H \rangle - E_0|)\) to evaluate how quickly and accurately each optimization method converges to the ground state energy. This provides insights into their efficiency and robustness. Parameter tuning was conducted using a grid search to identify the optimal settings for each optimizer.

\begin{figure}[H]
\centering
\includegraphics[width=\textwidth]{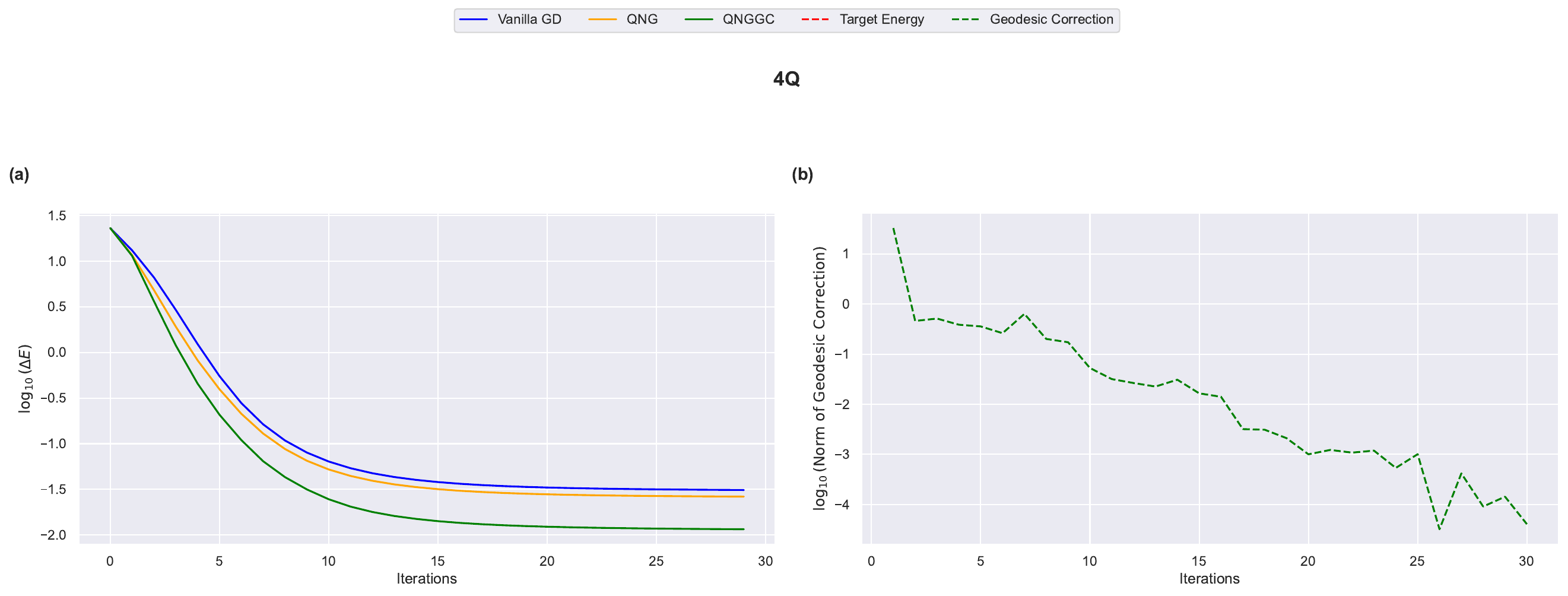}
\includegraphics[width=\textwidth]{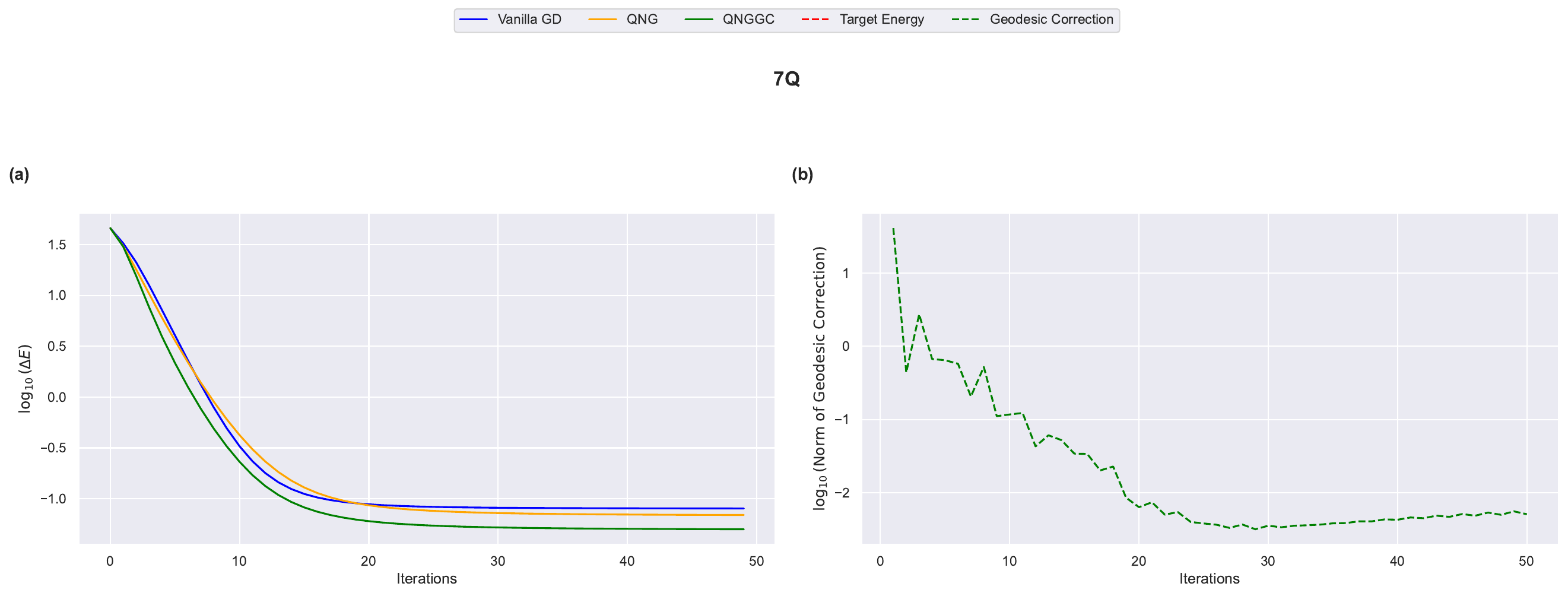}
\caption{
Convergence comparison of GD, QNG, and QNGGC for the Transverse Field Ising Model. The upper plot shows results for the 4-qubit system, and the lower plot for the 7-qubit system, averaged over 50 runs with random initializations. Subplot (a) plots \(\log_{10}(\Delta E)\), and (b) shows the norm of the geodesic correction term. Optimizers were run for 30 iterations for the 4-qubit system and 50 iterations for the 7-qubit system. The regularization was set to \(\lambda = 10^{-6}\).
}
\label{fig:exm3_4Q}
\end{figure}
\section{Conclusions and Outlook}\label{Conclusions}

In this manuscript, we extended the QNG method by incorporating higher-order integration techniques inspired by the Riemannian Euler update rule and a geometric perspective, leading to the development of an update rule based on geodesic equations. We introduced a tunable parameter \( b \) to heuristically adjust the effect of the geodesic correction, allowing for flexible integration of higher-order information. Through various examples, we demonstrated that QNGGC offers a more effective optimization strategy compared to traditional QNG by more closely following the direct optimization path. To compute the Christoffel symbols directly from quantum circuits, we proposed a method using the parameter-shift rule, enabling direct evaluation within quantum circuits.

While QNGGC offers improved convergence through geodesic corrections, it also introduces additional computational overhead. Gradient descent (GD) requires only \( O(l) \) circuit evaluations per iteration (with \( l \) denoting the number of variational parameters) and no metric-related overhead, but typically suffers from slow convergence. QNG improves convergence at a cost of \( O(l^2) \) circuits for evaluating the Fubini–Study metric. QNGGC further enhances convergence by incorporating curvature corrections via Christoffel symbols, increasing the cost to \( O(l^3) \). Although this complexity is manageable for small, shallow quantum circuits, it may become prohibitive in larger or noisier quantum circuits.

A promising direction for reducing this cost to a constant per iteration is the use of simultaneous perturbation stochastic approximation (SPSA)~\cite{Spall1992}, which is particularly well suited to such noisy, deep circuits, and has already demonstrated advantages for efficiently computing the Fubini–Study metric~\cite{Gacon}; a similar strategy could be applied to the computation of Christoffel symbols as well. In this work, we focused on small, shallow quantum circuits and the geometric structure underlying the rederivation of QNG from the Riemannian exponential map, as well as its extension to higher-order integrators that lead to QNGGC. We leave the exploration of deeper, noisy quantum circuits to future research. Once the computational cost is reduced through SPSA or other stochastic techniques, the framework could be extended further and find broader applications, such as time-dependent optimization methods inspired by differential geometry and the metric tensor formalism.

\section{Acknowledgements}

The author thanks Karl Jansen, Stefan Kühn, Yahui Chai, Yibin Guo, Tim Schwägerl, and Cenk Tüysüz from CQTA, DESY, as well as Tobias Hartung (Northeastern University, London) and Naoki Yamamoto (Keio University, Japan) for their helpful discussions and insightful feedback. The author also thanks Balázs Hetényi (Wigner Research Centre for Physics, Hungary) for his feedback. This work was supported by the Ministry of Science, Research, and Culture of the State of Brandenburg within the Centre for Quantum Technologies and Applications (CQTA) and funded by the German Ministry of Education and Research (BMBF) through the Project NiQ (Noise in Quantum Algorithms).

\begin{figure}[H]
\centering
\includegraphics[width=0.2\textwidth]{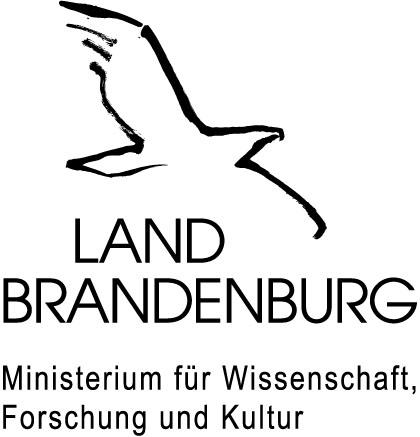}
\end{figure}


\appendix

\section{Parameter-Shift Rule and Metric Tensor Calculation}

This section provides a brief review of the parameter-shift rule and its application in computing both the gradient and the metric tensor for variational quantum algorithms. The parameter-shift rule allows for the direct estimation of gradients and higher-order derivatives on quantum hardware. For further details, we refer the reader to \cite{Mitarai, Schuld, Mari, Wierichs}.

\subsection{Parameter-Shift Rule for Gradient Estimation}

The parameter-shift rule provides a method to calculate the derivative of an expectation value with respect to a parameter \(\theta_i\) in a quantum circuit. Consider an observable \(\hat{O}\) and a parameterized quantum state \(|\psi(\bm{\theta})\rangle\); the expectation value of the observable is given by:
\begin{equation}
\langle \hat{O} \rangle(\bm{\theta}) = \langle \psi(\bm{\theta}) | \hat{O} | \psi(\bm{\theta}) \rangle.
\end{equation}

The derivative of this expectation value with respect to the parameter \(\theta_i\) can be computed using the parameter-shift rule as:

\begin{align}
\partial_{i} \langle \hat{O} \rangle(\bm{\theta}) &= \frac{1}{2} \left[\langle \hat{O} \rangle(\bm{\theta} + s \mathbf{e}_{i}) - \langle \hat{O} \rangle(\bm{\theta} - s \mathbf{e}_{i})\right],
\end{align}
where \(s\) is the shift in the parameter, typically \(s = \frac{\pi}{2}\), and \(\mathbf{e}_{i}\) is the unit vector in the direction of \(\theta_i\).

\subsection{Fubini-Study Metric}

To derive the Fubini-Study metric tensor \(g_{ij}(\bm{\theta})\), we begin by examining the infinitesimal distance between nearby quantum states in the parameter space. The metric tensor captures the local geometric structure of this space, reflecting how the quantum state \(|\psi(\bm{\theta})\rangle\) evolves under small displacements~\cite{Stokes2020, Meyer, Provost1980}.

We start by considering the normalization condition of the quantum state:
\begin{equation}
\langle \psi(\bm{\theta}) | \psi(\bm{\theta}) \rangle = 1.
\end{equation}

Differentiating this condition with respect to the parameters \(\theta^i\) gives:
\begin{equation}
\left\langle \psi(\bm{\theta})| \partial_i \psi(\bm{\theta}) \right\rangle 
+ \left\langle \partial_i \psi(\bm{\theta})| \psi(\bm{\theta}) \right\rangle = 0.
\label{eq:33}
\end{equation}

Differentiating Eq.~\eqref{eq:33} with respect to the parameters \(\theta^j\), we obtain:

\begin{equation}
\left\langle \psi(\bm{\theta})| \partial_i \partial_j \psi(\bm{\theta}) \right\rangle 
+ \left\langle \partial_i \partial_j \psi(\bm{\theta})| \psi(\bm{\theta}) \right\rangle 
+ \left\langle \partial_i \psi(\bm{\theta})| \partial_j \psi(\bm{\theta}) \right\rangle 
+ \left\langle \partial_j \psi(\bm{\theta})| \partial_i \psi(\bm{\theta}) \right\rangle = 0.
\label{eq:34}
\end{equation}

Next, consider the quantum state \(|\psi(\bm{\theta} + \boldsymbol{\delta\theta})\rangle\) near \(|\psi(\bm{\theta})\rangle\), where \(\boldsymbol{\delta\theta}\) is the displacement vector. Using a Taylor expansion, we express the shifted state as:
\begin{equation}
|\psi(\bm{\theta} + \delta\theta)\rangle = |\psi(\bm{\theta})\rangle + \partial_i |\psi(\bm{\theta})\rangle \delta\theta^i + \frac{1}{2} \partial_i \partial_j |\psi(\bm{\theta})\rangle \delta\theta^i \delta\theta^j + \mathcal{O}(\delta\theta^3).
\end{equation}

Taking the inner product of \(|\psi(\bm{\theta})\rangle\) with this expanded state, we get:
\begin{equation}
\langle \psi(\bm{\theta}) | \psi(\bm{\theta} + \bm{\delta\theta}) \rangle = 1 + \langle \psi(\bm{\theta}) | \partial_i \psi(\bm{\theta}) \rangle \delta\theta^i + \frac{1}{2} \langle \psi(\bm{\theta}) | \partial_i \partial_j \psi(\bm{\theta}) \rangle \delta\theta^i \delta\theta^j + \mathcal{O}(\delta\theta^3).
\end{equation}

The fidelity between the states \(|\psi(\bm{\theta})\rangle\) and \(|\psi(\bm{\theta} + \delta\theta)\rangle\) is given to quadratic order in $\bm{\delta \theta}$ by:
\begin{align}
|\langle \psi(\bm{\theta})| \psi(\bm{\theta} + \bm{\delta\theta}) \rangle|^2 
&= \langle \psi(\bm{\theta})| \psi(\bm{\theta} + \bm{\delta\theta}) \rangle \langle \psi(\bm{\theta} + \bm{\delta\theta})| \psi(\bm{\theta}) \rangle \nonumber \\
&= 1 + \left[ \left\langle \psi(\bm{\theta})| \partial_i \psi(\bm{\theta}) \right\rangle + \left\langle \partial_i \psi(\bm{\theta})| \psi(\bm{\theta}) \right\rangle \right] \delta\theta^i \nonumber \\
&\quad + \left[ \left\langle \partial_i \psi(\bm{\theta})|  \psi(\bm{\theta}) \right\rangle  \left\langle \psi(\bm{\theta}) | \partial_j \psi(\bm{\theta}) \right\rangle \right] \delta\theta^i \delta\theta^j \nonumber \\
&\quad + \frac{1}{2} \left[ \left\langle \psi(\bm{\theta}) | \partial_i \partial_j \psi(\bm{\theta}) \right\rangle + \left\langle \partial_i \partial_j \psi(\bm{\theta}) | \psi(\bm{\theta}) \right\rangle \right] \delta\theta^i \delta\theta^j \nonumber + \cdots \\
\end{align}
Using Eqs.~\eqref{eq:33} and \eqref{eq:34}, we get
\begin{align}
|\langle \psi(\bm{\theta})| \psi(\bm{\theta} + \bm{\delta\theta}) \rangle|^2
&= 1 + \bigg[ 
\left\langle \partial_i \psi(\bm{\theta})| \psi(\bm{\theta}) \right\rangle 
\left\langle \psi(\bm{\theta}) | \partial_j \psi(\bm{\theta}) \right\rangle \nonumber \\
&\quad 
- \frac{1}{2} \left( 
\left\langle \partial_i \psi(\bm{\theta})| \partial_j \psi(\bm{\theta}) \right\rangle 
+ \left\langle \partial_j \psi( \bm{\theta})| \partial_i \psi(\bm{\theta}) \right\rangle 
\right) 
\bigg] \delta\theta^i \delta\theta^j + \cdots
\end{align}

Using the approximation for the Fubini-Study distance between quantum states,
\begin{equation}
d^2(P_\psi, P_\phi) = \arccos^2(|\langle \psi | \phi \rangle|) \approx 1 - |\langle \psi | \phi \rangle|^2 + \mathcal{O}((1 - |\langle \psi | \phi \rangle|^2)^2),
\end{equation}
and the fidelity expansion up to second order, it follows that the infinitesimal squared distance is given by
\begin{align}
d^2(P_{\psi(\bm{\theta})}, P_{\psi(\bm{\theta} + \bm{\delta\theta})}) 
&= \left[
\frac{1}{2} \left( 
\left\langle \partial_i \psi(\bm{\theta})| \partial_j \psi(\bm{\theta}) \right\rangle 
+ \left\langle \partial_j \psi(\bm{\theta})| \partial_i \psi(\bm{\theta}) \right\rangle 
\right) \right. \nonumber \\
&\quad \left. - 
\left\langle \partial_i \psi(\bm{\theta})| \psi(\bm{\theta}) \right\rangle 
\left\langle \psi(\bm{\theta})| \partial_j \psi(\bm{\theta}) \right\rangle 
\right] d \theta^i d \theta^j. \label{eq:qgt_expanded}
\end{align}

The first term in \eqref{eq:qgt_expanded} is manifestly real. The second term is also real because from the normalization condition (as shown in Eq.~\eqref{eq:33}), we have
\begin{equation}
\text{Re}\left[ \left\langle \psi(\bm{\theta})| \partial_i \psi(\bm{\theta}) \right\rangle \right] = 0.
\end{equation}

Therefore, the infinitesimal squared distance is given by the real part of the quantum geometric tensor:
\begin{align}
d^2(P_{\psi(\bm{\theta})}, P_{\psi(\bm{\theta} + \bm{\delta\theta})}) 
&= \text{Re} \left[
\left\langle \partial_i \psi(\bm{\theta})| \partial_j \psi(\bm{\theta}) \right\rangle 
- \left\langle \partial_i \psi(\bm{\theta})| \psi(\bm{\theta}) \right\rangle 
\left\langle \psi(\bm{\theta})| \partial_j \psi(\bm{\theta}) \right\rangle 
\right] d \theta^i d \theta^j. \label{eq:fs_distance}
\end{align}

This directly defines the Fubini-Study metric as:
\begin{equation}
g_{ij}(\bm{\theta}) = \text{Re} \left[
\left\langle \partial_i \psi(\bm{\theta}) | \partial_j \psi(\bm{\theta}) \right\rangle 
- \left\langle \partial_i \psi(\bm{\theta}) | \psi(\bm{\theta}) \right\rangle 
\left\langle \psi(\bm{\theta}) | \partial_j \psi(\bm{\theta}) \right\rangle 
\right]. \label{eq:fs_metric}
\end{equation}

Thus, the Fubini-Study metric tensor \(g_{ij}(\bm{\theta})\) quantifies the infinitesimal squared distance between nearby quantum states in the parameter space, providing a geometric interpretation of the state evolution. To express this metric tensor in terms of measurable quantities, we employ the parameter-shift rule, which enables direct computation on quantum hardware. The resulting expressions for the metric tensor components are given in~\eqref{full_metric} for the full form and~\eqref{diag_metric} for the diagonal approximation.

\subsection{Proof of the Proposition}\label{app:Christoffel}

Using the higher-order derivatives strategy with the parameter-shift rule, the differentiation of the metric \eqref{diag_metric} with respect to \(\theta_j\) is given by:

\begin{align}
\partial_k g_{j,j}(\boldsymbol{\theta}) = \frac{1}{8} \Bigg[
- \left| \langle \psi(\boldsymbol{\theta}) | \psi(\boldsymbol{\theta} + \pi \mathbf{e}_j + \frac{\pi}{2} \mathbf{e}_k) \rangle \right|^2 
+ \left| \langle \psi(\boldsymbol{\theta}) | \psi(\boldsymbol{\theta} + \pi \mathbf{e}_j - \frac{\pi}{2} \mathbf{e}_k) \rangle \right|^2
\Bigg].
\label{diag_derv}
\end{align}

Using the definition of the Christoffel symbols \eqref{DefChristoffel}, for a diagonal metric tensor, the Christoffel symbols are determined by the following specific cases:

1. For \( i = j = k \):
   \[
   \Gamma^i_{ii} = \frac{1}{2 g_{ii}} \frac{\partial g_{ii}}{\partial \theta^i},
   \]

2. For \( i = j \neq k \):
   \[
   \Gamma^i_{ik} = \frac{1}{2 g_{ii}} \frac{\partial g_{ii}}{\partial \theta^k},
   \]

3. For \( i = k \neq j \):
   \[
   \Gamma^i_{ij} = \frac{1}{2 g_{ii}} \frac{\partial g_{ii}}{\partial \theta^j},
   \]

4. For \( j = k \neq i \):
   \[
   \Gamma^i_{jj} = -\frac{1}{2 g_{ii}} \frac{\partial g_{jj}}{\partial \theta^i},
   \]
   
5. For \( i \neq j \neq k \):
   \[
   \Gamma^i_{jk} = 0.
   \]
   
The case \( i = j = k \) is ignored as shifting the same parameter twice disrupts distinct mixed derivative calculations.

Using the Kronecker delta notation, the expression can be summarized as:

\begin{align}
\Gamma^i_{jk} = \frac{1}{2 g_{ii}} \left( \delta_{ij} \frac{\partial g_{ii}}{\partial \theta^k} + \delta_{ik} \frac{\partial g_{ii}}{\partial \theta^j} - \delta_{jk} \frac{\partial g_{jj}}{\partial \theta^i} \right),
\label{kroneker_proof}
\end{align}
where \( \delta_{ij} \) is the Kronecker delta, defined as:
\[
\delta_{ij} = \begin{cases}
1 & \text{if } i = j, \\
0 & \text{if } i \neq j.
\end{cases}
\]

By substituting equation \eqref{diag_derv} into equation \eqref{kroneker_proof}, we obtain the final expression for the Christoffel symbols for the diagonal metric \eqref{Christoffel_parshiff}. 
\(\square\)


\section{Detail calculation for Example 2}
\label{appendix:metric-christoffel}

In this appendix, we provide the detailed expressions for the inverse Fubini-Study metric and the Christoffel symbols for Example 2, which simulates the Hydrogen molecule (\(\text{H}_2\)).

\subsection{Inverse Fubini-Study Metric}
The inverse Fubini-Study metric \( F^{-1} \), after applying Tikhonov regularization, is given by:

\[
\scalebox{0.67}{$
\left(
\begin{array}{ccc}
\frac{3 + 4\lambda (3 + 2\lambda) + \cos(4\theta_0) - 4(1 + \lambda)\cos(2\theta_0)\cos(2\theta_1)}{(1 + \lambda)(2 + 4\lambda (3 + 2\lambda) + \cos(4\theta_0) - 4(1 + \lambda)\cos(2\theta_0)\cos(2\theta_1) + \cos(4\theta_1))} & 
- \frac{8\cos(\theta_0)\cos(\theta_1)\sin(\theta_0)\sin(\theta_1)}{(1 + \lambda)(2 + 4\lambda (3 + 2\lambda) + \cos(4\theta_0) - 4(1 + \lambda)\cos(2\theta_0)\cos(2\theta_1) + \cos(4\theta_1))} & 
-\frac{4\sin(2\theta_1)}{2 + 4\lambda (3 + 2\lambda) + \cos(4\theta_0) - 4(1 + \lambda)\cos(2\theta_0)\cos(2\theta_1) + \cos(4\theta_1)} \\[10pt]

-\frac{8\cos(\theta_0)\cos(\theta_1)\sin(\theta_0)\sin(\theta_1)}{(1 + \lambda)(2 + 4\lambda (3 + 2\lambda) + \cos(4\theta_0) - 4(1 + \lambda)\cos(2\theta_0)\cos(2\theta_1) + \cos(4\theta_1))} & 
\frac{1 + \frac{2\sin(2\theta_0)^2}{2 + 4\lambda (3 + 2\lambda) + \cos(4\theta_0) - 4(1 + \lambda)\cos(2\theta_0)\cos(2\theta_1) + \cos(4\theta_1)}}{1 + \lambda} & 
\frac{8\cos(\theta_0)\sin(\theta_0)}{2 + 4\lambda (3 + 2\lambda) + \cos(4\theta_0) - 4(1 + \lambda)\cos(2\theta_0)\cos(2\theta_1) + \cos(4\theta_1)} \\[10pt]

-\frac{4\sin(2\theta_1)}{2 + 4\lambda (3 + 2\lambda) + \cos(4\theta_0) - 4(1 + \lambda)\cos(2\theta_0)\cos(2\theta_1) + \cos(4\theta_1)} & 
-\frac{8\cos(\theta_0)\sin(\theta_0)}{2 + 4\lambda (3 + 2\lambda) + \cos(4\theta_0) - 4(1 + \lambda)\cos(2\theta_0)\cos(2\theta_1) + \cos(4\theta_1)} & 
\frac{8(1 + \lambda)}{2 + 4\lambda (3 + 2\lambda) + \cos(4\theta_0) - 4(1 + \lambda)\cos(2\theta_0)\cos(2\theta_1) + \cos(4\theta_1)}
\end{array}
\right)
$}
\]

\subsection{Christoffel Symbols}

The Christoffel symbols $\Gamma^{i}_{jk}$ for this example are given by:

\begin{figure}[H]
\centering
\includegraphics[width=1.1\textwidth]{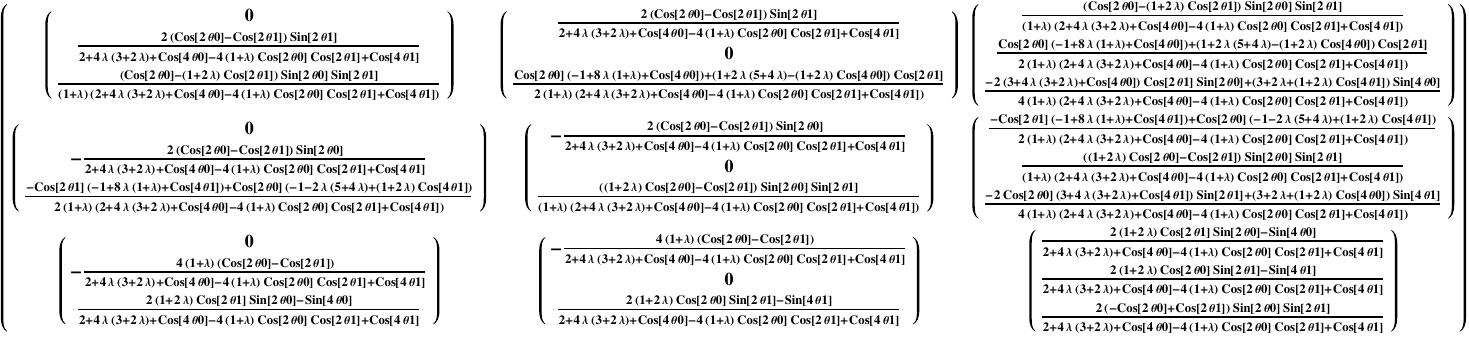}
\end{figure}

where the Christoffel symbols $\Gamma^i_{jk}$ were generated using \textit{Mathematica}~\cite{Mathematica} and are presented here exactly as printed by the software, in compact matrix form:
\[
\left(
\begin{array}{ccc}
\begin{pmatrix}
\Gamma^0_{00} \\
\Gamma^0_{10} \\
\Gamma^0_{20}
\end{pmatrix}
&
\begin{pmatrix}
\Gamma^0_{01} \\
\Gamma^0_{11} \\
\Gamma^0_{21}
\end{pmatrix}
&
\begin{pmatrix}
\Gamma^0_{02} \\
\Gamma^0_{12} \\
\Gamma^0_{22}
\end{pmatrix}
\\[2.5em]
\begin{pmatrix}
\Gamma^1_{00} \\
\Gamma^1_{10} \\
\Gamma^1_{20}
\end{pmatrix}
&
\begin{pmatrix}
\Gamma^1_{01} \\
\Gamma^1_{11} \\
\Gamma^1_{21}
\end{pmatrix}
&
\begin{pmatrix}
\Gamma^1_{02} \\
\Gamma^1_{12} \\
\Gamma^1_{22}
\end{pmatrix}
\\[2.5em]
\begin{pmatrix}
\Gamma^2_{00} \\
\Gamma^2_{10} \\
\Gamma^2_{20}
\end{pmatrix}
&
\begin{pmatrix}
\Gamma^2_{01} \\
\Gamma^2_{11} \\
\Gamma^2_{21}
\end{pmatrix}
&
\begin{pmatrix}
\Gamma^2_{02} \\
\Gamma^2_{12} \\
\Gamma^2_{22}
\end{pmatrix}
\end{array}
\right)
\]

\noindent Mathematica formats functions symbolically, such as \texttt{Cos[$\theta a$]} for $\cos(\theta_a)$ and \texttt{Sin[$\theta a$]} for $\sin(\theta_a)$, where $a = 0, 1$. We retain this notation to preserve the symbolic structure and ensure consistency with the software's output.

\subsection{Ansatz, Cost Function, and Gradient}

The ansatz given in equation \eqref{ansatz_ex2} can be expanded as follows:

\begin{equation}
\begin{aligned}
|\psi\rangle 
&= \left(\cos(\theta_0) \cos(\theta_1)\right)\ket{00} \\
&\quad + \left(\cos(\theta_0) \cos(\theta_2) \sin(\theta_1) - \cos(\theta_1) \sin(\theta_0) \sin(\theta_2)\right)\ket{01} \\
&\quad + \left(\sin(\theta_0) \sin(\theta_1)\right)\ket{10} \\
&\quad + \left(\cos(\theta_1) \cos(\theta_2) \sin(\theta_0) + \cos(\theta_0) \sin(\theta_1) \sin(\theta_2)\right)\ket{11}.
\end{aligned}
\end{equation} 

Using this ansatz and the hydrogen, $H_2$, Hamiltonian, we derive the following cost function:

\begin{equation}
\begin{aligned}
\mathcal{L}(\theta_0, \theta_1, \theta_2, \theta_3) &= 2 \cos^2(\theta_0) \left(-\alpha \sin^2(\theta_1) \sin^2(\theta_2) + \alpha \cos^2(\theta_1) + \beta \sin(\theta_1) \cos(\theta_1) \sin(\theta_2)\right) \\
&\quad - 2 \sin^2(\theta_0) \cos(\theta_1) \left(\alpha \cos(\theta_1) \cos^2(\theta_2) + \beta \sin(\theta_1) \sin(\theta_2)\right) \\
&\quad + 2 \sin(\theta_0) \cos(\theta_0) \cos(\theta_2) (\beta - 2 \alpha \sin(\theta_1) \cos(\theta_1) \sin(\theta_2))
\end{aligned}
\end{equation}

The gradient of the cost function with respect to the parameters\( \theta_0 \), \( \theta_1 \), and \( \theta_2 \) is expressed as:

\begin{equation}
\begin{aligned}
\frac{\partial \mathcal{L}}{\partial \theta_0} &= -\sin(2\theta_0) \left(\alpha + 2\alpha \cos(2\theta_1) + \alpha \cos(2\theta_2) + 2\beta \sin(2\theta_1) \sin(\theta_2)\right) \\
&\quad + \cos(2\theta_0) \left(2\beta \cos(\theta_2) - \alpha \sin(2\theta_1) \sin(2\theta_2)\right), \\
\frac{\partial \mathcal{L}}{\partial \theta_1} &= \cos(2\theta_0) \left(-2\alpha \sin(2\theta_1) + 2\beta \cos(2\theta_1) \sin(\theta_2)\right) \\
&\quad - \alpha \left(2 \sin(2\theta_1) \sin^2(\theta_2) + \cos(2\theta_1) \sin(2\theta_0) \sin(2\theta_2)\right), \\
\frac{\partial \mathcal{L}}{\partial \theta_2} &= \beta \cos(2\theta_0) \cos(\theta_2) \sin(2\theta_1) \\
&\quad - \sin(2\theta_0) \left(\alpha \cos(2\theta_2) \sin(2\theta_1) + \beta \sin(\theta_2)\right) + \alpha \left(-\cos(2\theta_0) + \cos(2\theta_1)\right) \sin(2\theta_2).
\end{aligned}
\end{equation}


\end{document}